\documentclass[oneside]{amsart}
\usepackage{geometry}
\usepackage{enumerate}[\indent 1.]
\usepackage{mathtools}
\usepackage{amsfonts}
\usepackage[comma, sort&compress]{natbib}
\usepackage{soul}
\usepackage{amsthm}
\usepackage{multirow}
\usepackage[utf8]{inputenc}
\usepackage{setspace}
\usepackage{subfig}
\usepackage{leftidx}
\usepackage{mathabx}
\usepackage{varwidth}
\usepackage{xcolor}
\usepackage{mathrsfs}
\usepackage{booktabs}
\usepackage{breqn}
\usepackage[linesnumbered, ruled, noline]{algorithm2e}
\usepackage{hyperref}
\usepackage[htt]{hyphenat}
\usepackage{float} % for H option
\usepackage{booktabs} % for \resizebox
\usepackage{tabularx} % for tabularx environment
\usepackage{ulem}  % Required for inserting images

\usepackage[T1]{fontenc}
\usepackage[utf8]{inputenc}
\usepackage{babel}
\usepackage[font=small,labelfont=bf]{caption}

\usepackage{paralist}
\RequirePackage{hyperref}

\usepackage{todonotes}

\makeatletter
\newcommand\dotover{\leavevmode\cleaders\hb@xt@ .22em{\hss $\cdot$\hss}\hfill\kern\z@}

\makeatother

\newtheorem{theorem}{Theorem}
\newtheorem{assumption}{Assumption}

\newtheorem{lemma}[theorem]{Lemma}

\newtheorem{corollary}[theorem]{Corollary}

\newtheorem{definition}{Definition}
\newtheorem{claim*}{Claim}

\numberwithin{equation}{section}
\numberwithin{theorem}{section}
\numberwithin{example}{section}
\numberwithin{table}{section}
\numberwithin{figure}{section}

\newcommand{\assumpref}[1]{Assumption~\ref{assump:#1}}

\newcommand{\figref}[1]{Figure~\ref{fig:#1}}

\newcommand{\secref}[1]{Section~\ref{sec:#1}}

\newcommand{\appref}[1]{Appendix~\ref{app:#1}}

\newcommand{\defref}[1]{Definition~\ref{def:#1}}

\newcommand{\lemref}[1]{Lemma~\ref{lem:#1}}

\newcommand{\thmref}[1]{Theorem~\ref{thm:#1}}

\newcommand{\corref}[1]{Corollary~\ref{cor:#1}}

\newcommand{\FWER}{\text{FWER}}

\allowdisplaybreaks

\title{An adaptive procedure for detecting replicated signals with $k$-family-wise error rate control}
\author{Ninh Tran}
\keywords{$k$-FWER, multiple testing, partial conjunction, adaptive, filtering.}

\thanks{This research was supported by The University of Melbourne's Research Computing Services, the Petascale Campus Initiative, and the Australian Government Research Training Program (RTP) Scholarship.}

\subjclass[2020]{62F03}

\begin{document}

\maketitle

\begin{abstract}
   Partial conjunction (PC) hypothesis testing is widely used to assess the replicability of scientific findings across multiple comparable studies. In high-throughput meta-analyses, testing a large number of PC hypotheses with $k$-family-wise error rate ($k$-FWER) control often suffers from low statistical power due to the multiplicity burden. The state-of-the-art AdaFilter-Bon procedure by \citet[Ann. Stat., 50(4), 1890-1909]{Wang2022} alleviates this problem by filtering out hypotheses unlikely to be false before applying a rejection rule. However, a side effect of filtering is that it renders the rejection rule more stringent than necessary, leading to conservative $k$-FWER control. In this paper, we mitigate this conservativeness—and thereby improve the power of AdaFilter-Bon—by incorporating a post-filter null proportion estimate into the procedure. The resulting method, AdaFilter-AdaBon, has proven asymptotic $k$-FWER control under weak dependence and demonstrates empirical finite-sample control with higher power than the original AdaFilter-Bon in simulations.
\end{abstract}

\section{Introduction} 
Over the last two decades, the partial conjunction (PC) hypothesis testing \citep{Benjamini2008} has become a standard statistical tool for assessing replicability in meta-analyses involving $n \geq 2$ comparable studies. Within this framework, a scientifically interesting finding (i.e., a \textit{signal}) is deemed credible if it is exhibited in at least $u$ of the $n$ studies, a notion referred to as $u/n$ \textit{replication}. The value of $u \in \{ 2, \dots, n \}$ is freely chosen by the meta-analyst and is often guided by conventions within the field. For example, $u = 2$ is widely regarded as the gold standard in drug assessment \citep{Zhan2022}. 

Identifying $u/n$ replication requires testing the PC null hypothesis that \textit{fewer} than $u$ of the $n$ studies exhibit the signal. A $p$-value for testing this null hypothesis is typically obtained by combining the $p$-values from the $n$ individual studies using the general methodology of \citet{Benjamini2008}, which was later extended by \citet{Wang2019}. Testing many PC nulls with combined $p$-values under \textit{family-wise error rate} (FWER) control often suffers from low power due to the multiplicity correction. This issue is especially acute in high-throughput meta-analyses of gene expression, where the number of PC nulls can run into the hundreds or thousands. \citet{Wang2022} address this problem by proposing \textbf{AdaFilter-Bon}, a state-of-the-art FWER method that uses a $p$-value-\textbf{ada}ptive \textbf{filter} to remove PC nulls that are likely to be true, before applying a rejection rule similar to the \textbf{Bon}ferroni procedure.

In \secref{kFWER_methods}, we provide an overview of AdaFilter-Bon and show that, with a simple modification, AdaFilter-Bon can control the $k$-FWER—a generalization of the FWER\footnote{When $k = 1$, the $k$-FWER is equivalent to the FWER.}—for any integer $k \in \mathbb{N}$. While filtering reduces the multiplicity burden and increases power, it can also result in overly conservative $k$-FWER control. For a target $k$-FWER level $\alpha \in (0,1)$, we show in \secref{preview} that AdaFilter-Bon controls the $k$-FWER at $\alpha$ multiplied by the \textit{expected proportion of true PC nulls among those remaining after filtering}. This proportion, which lies in $[0,1]$, decreases as more true PC nulls are filtered out. Consequently, when the filter operates ideally—removing many PC nulls likely to be true—AdaFilter-Bon’s $k$-FWER level can drop well below $\alpha$.

In this paper, we introduce AdaFilter-\textbf{Ada}Bon, a multiple testing procedure based on AdaFilter-Bon that utilises a $p$-value-\textbf{ada}ptive estimator of the post-filter proportion of true PC nulls to offset the conservative $k$-FWER control induced by filtering. Under mild assumptions that allow for weak dependence among the data, this new procedure has theoretical asymptotic $k$-FWER control. In simulation studies, we show that AdaFilter-AdaBon maintains robust finite-sample $k$-FWER control and yields more true positive discoveries than the original AdaFilter-Bon.

To the best of our knowledge, this is the first work to employ an estimate of the post-filter proportion of true PC nulls for controlling the $k$-FWER. In particular, we were unable to find existing research addressing $k$-FWER control in replicability analyses, and found only limited efforts to estimate the post-filter proportion of true PC nulls \citep{Dickhaus2024, Tran2025, Bogolomov2018}. Even in single-study multiple testing, research on $k$-FWER control has been modest, with notable contributions from \citet{Lehmann2005, Romano2006, Sarkar2007, Sarkar2008}, none of which incorporate null proportion estimation.

The structure of the paper is as follows. In \secref{problem}, we present the problem formulation, overview the AdaFilter-Bon procedure, and outline the theory behind its conservative $k$-FWER control. In \secref{AdaFilterAdaBon}, we introduce the new method, AdaFilter-AdaBon, and establish its asymptotic $k$-FWER control guarantees. \secref{num} reports a simulation study comparing the performance of AdaFilter-AdaBon with AdaFilter-Bon and other related methods. Finally, \secref{discussion} concludes the paper by highlighting extensions of AdaFilter-AdaBon and open challenges for future work.

\section{Problem formulation}\label{sec:problem}
\subsection{Problem statement}\label{sec:problem_statement}
Suppose we are conducting a meta-analysis that consists of $n \geq 2$ comparable studies, each of which examines the same set of $m > 1$ features. For each feature $i \in \{ 1, \dots, m \}$ of study $j \in \{ 1, \dots, n \}$, suppose there is a null hypothesis $H_{ij}$ for which we observe a corresponding \textit{valid} $p$-value $P_{ij}$, i.e.,
\begin{equation*}
    \text{$\Pr(P_{ij} \leq t) \leq t$ for all $t \in [0,1]$ if $H_{ij}$ is true}.
\end{equation*}
Since a meta-analysis typically consists of independently conducted studies, we assume that the vectors $(P_{i1})^m_{i=1}$, $\dots$, $(P_{in})^m_{i=1}$ are independent.

For a given \textit{replicability level} $u \in \{ 2, \dots, n \}$, feature $i$ is said to be $u/n$ \textit{replicated} if at least $u$ of the $n$ hypotheses among $H_{i1},\dots,H_{in}$ are false. We consider a problem where the objective is to identify $u/n$ replicated features by testing the partial conjunction (PC) null hypothesis \citep{Benjamini2008}:
\begin{equation*}
    H^{u/n}_i : \text{ Less than $u$ null hypotheses among $H_{i1},\dots,H_{in}$ are false},
\end{equation*}
for $i = 1, \dots, m$. This multiple testing problem is often referred to as a \textit{replicability analysis}. Note that whenever $u'' \geq u'$, the following nesting property holds:
\begin{equation}\label{nest}
    \text{$H_i^{u'/n}$ is true} \implies \text{$H_i^{u''/n}$ is true}.
\end{equation}

Let $\mathcal{R} \subseteq \{1, \dots, m\}$ be the set of PC null hypotheses rejected by a multiple testing procedure. For any given tolerance level $k \in \mathbb{N}$, the $k$-family-wise error rate ($k$-FWER) of $\mathcal{R}$ is defined as
\begin{equation*}
    k{\text -}\FWER \equiv \Pr\left( \sum^m_{i=1} I\{ i \in \mathcal{R} \} I \{ \text{$H^{u/n}_i$ is true}  \} \geq k \right),
\end{equation*}
which is the probability that the number of false discoveries is greater than or equal to $k$. When $k = 1$, the $k$-FWER is simply referred to as the $\FWER$. The power of $\mathcal{R}$ can be evaluated using the true positive rate (TPR),
\begin{equation*}
    \text{TPR} \equiv \mathbb{E}\left[ \frac{\sum^m_{i=1} I \{ i \in \mathcal{R} \} I\{ \text{$H^{u/n}_i$ is false} \} }{1 \vee \sum^m_{i=1} I\{\text{$H^{u/n}_i$ is false} \} }  \right],
\end{equation*}
which is the expected proportion of true discoveries among the false PC nulls. Our goal is to develop a powerful multiple testing procedure for replicability analysis that operates on 
\begin{equation}\label{Ps}
  (P_{ij})_{m \times n} \equiv (P_{ij})_{(i,j) \in \{ 1 , \dots, m \} \times \{1, ,\dots, n \} },  
\end{equation}
and controls the $k$-FWER below a pre-specified target level $\alpha \in (0,1)$.
 
\subsection{$k$-FWER procedures for replicability analysis}\label{sec:kFWER_methods}
A valid $p$-value for testing $H^{u/n}_i$—referred to as a \textit{PC $p$-value} and denoted as $P^{u/n}_i$—can be constructed by \textit{combining} the $p$-values associated with feature $i$. That is,
\begin{equation*}
    P^{u/n}_i \equiv f(P_{i1}, P_{i2}, \dots, P_{in}; u),
\end{equation*}
where $f: [0,1]^n \to [0,1]$ is a combining function satisfying
\begin{equation*}
    \Pr ( f(P_{i1}, P_{i2}, \dots, P_{in}) \leq t ) \leq t \quad \text{ for all $t \in [0,1]$ if $H^{u/n}_i$ is true.} 
\end{equation*} 
\citet{Benjamini2008} proposed a general methodology for constructing a valid combined $p$-value for testing $H^{u/n}_i$, which \citet{Wang2019} later extended. An example of a $P^{u/n}_i$ from this methodology is the Bonferroni-combined PC $p$-value, 
\begin{equation*}
   f_{\text{Bon}}(P_{i1},\dots,P_{in}; u) = (n - u + 1) P_{i(u)},
\end{equation*}
where $P_{i(1)} \leq \cdots \leq P_{i(n)}$ are the order statistics of $P_{i1}, \dots, P_{in}$. Further examples and their numerical performance can be found in \citet{Hoang2022}.

Given the problem statement in \secref{problem_statement}, the $k$-FWER can be controlled by applying a multiple testing procedure on the PC $p$-values $(P^{u/n}_i)_{i=1}^m$. Arguably, the simplest $k$-FWER method is the (generalised) Bonferroni procedure introduced in \citet{Lehmann2005}:
\begin{definition}[Bonferroni]\label{def:Bon}
    Let $\alpha \in [0,1]$ be a $k$-FWER target for a given $k \in \mathbb{N}$, and $(P^{u/n}_i)^m_{i = 1}$ be PC $p$-values. For $i = 1, \dots, m$, reject $H^{u/n}_i$ if
    \begin{equation}\label{bon}
        P^{u/n}_i \leq k \cdot \frac{\alpha}{m}.
    \end{equation}
\end{definition}
Although less powerful than the subsequent methods developed by \citet{Romano2006, Sarkar2007, Sarkar2008}, the Bonferroni procedure provides a useful illustration of the multiplicity burden: as the number of PC nulls $m$ grows, the rejection threshold in \eqref{bon} becomes increasingly stringent.

The state-of-the-art AdaFilter-Bon procedure \citep{Wang2022} reduces the multiplicity burden by using a \textit{filter} to eliminate features unlikely to be $u/n$ replicated, before applying a rejection rule. Although originally developed for FWER control, we describe below a generalisation of AdaFilter-Bon that controls $k$-FWER control for any $k \in \mathbb{N}$\footnote{When $k = 1$, \defref{AdafilterBon} recovers the original AdaFilter-Bon procedure of \citet[Definition 3.1]{Wang2022}.}.
\begin{definition}[AdaFilter-Bon]\label{def:AdafilterBon}
     Let $\alpha \in [0,1]$ be a $k$-FWER target for a given $k \in \mathbb{N}$, and $u \in \{2, \dots, m\}$ be a replicability level. For $i = 1,\dots, m$, let
    \begin{equation}\label{S_i}
         S_i \equiv f_{\text{Bon}}(P_{i1},\dots,P_{in}; u)
    \end{equation}
    be a Bonferroni-combined PC $p$-value, and let
\begin{align}
    F_i \equiv \frac{n - u + 1}{n - u + 2} \cdot  f_{\text{Bon}}(P_{i1},\dots,P_{in}; u - 1) = (n - u + 1) P_{i(u-1)} \label{F_i}
\end{align}
be a corresponding \textit{filtering} $p$-value. For each $i$, reject $H^{u/n}_i$ if $S_i < \hat{t}$, where
    \begin{equation}\label{t_tilde}
        \hat{t} \equiv \sup \left \{ t \in [0, k  \alpha] : t \cdot 
        \sum^m_{i = 1} I \{ F_i < t \} \leq  k \alpha \right\}.
    \end{equation}
\end{definition}

Since $F_i \leq S_i$, it follows from \defref{AdafilterBon} that $H^{u/n}_i$ is filtered out (i.e. cannot be rejected) if $F_i \in [\hat{t},1]$. This is a sensible filtering approach, since $F_i$ is closely related to $f_{\text{Bon}}(P_{i1}, \dots, P_{in}; u-1)$, which is a valid PC $p$-value for testing $H_i^{(u-1)/n}$. Consequently, $F_i \in [\hat{t},1]$ indicates that $H_i^{(u-1)/n}$ is likely true, making feature $i$ unlikely to be $u/n$ replicated by the nesting property in \eqref{nest}.

We also observe that AdaFilter-Bon can only reject $H^{u/n}_i$ if
\begin{equation}\label{alt}
    S_i < \hat{t} \leq k \cdot \frac{\alpha}{\sum^m_{\ell = 1} I \{ F_{\ell} < \hat{t} \} }.
\end{equation}
The right-hand side of \eqref{alt} shows that $\hat{t}$ is inversely proportional to the number of features remaining after filtering, $\sum_{\ell=1}^m I\{F_{\ell} < \hat{t}\}$, rather than to $m$. Hence, in comparison to the Bonferroni procedure (\defref{Bon}), there is a reduction in the multiplicity burden due to filtering.

\subsection{Preview of contribution}\label{sec:preview}
By design, features not filtered out by AdaFilter-Bon, $\{ \ell \in \{1, \dots, m \} : F_{\ell} < \hat{t} \}$, are more likely to correspond to false PC nulls than those filtered out. Consequently, there is a tendency for the \textit{post-filter null proportion}, 
\begin{equation}\label{pi0}
        \pi_0(\hat{t}) \equiv \frac{\sum^m_{i = 1} I \{ F_i < \hat{t} \} I \{ \text{$H^{u/n}_i$ is true} \} }{\sum^m_{i = 1} I \{ F_i < \hat{t} \} } \in [0,1],
\end{equation}
to be small. This phenomenon, together with the theorem below, reveals that AdaFilter-Bon has a propensity to control the $k$-FWER conservatively:
\begin{theorem}[AdaFilter-Bon $k$-FWER control]\label{thm:AdaBon_FWER_control}
    If $(P_{ij})_{m \times n}$ is a collection of independent valid $p$-values, then AdaFilter-Bon (\defref{AdafilterBon}) has the following $ k\text{-FWER}$ property:
    \begin{equation}\label{AdafilterBonControl}
       k\text{-FWER}(\mathcal{R}) \leq \alpha \cdot \mathbb{E} \left[ \pi_0(\hat{t}) \right] \in [0,\alpha],
    \end{equation}
    where $\mathcal{R} = \{ i \in \{1 ,\dots, m \}: S_i \leq \hat{t} \}$ and $\pi_0(\hat{t})$ is as defined \eqref{pi0}.
\end{theorem} 
\thmref{AdaBon_FWER_control} highlights that AdaFilter-Bon’s $k$-FWER level can fall well below $\alpha$ when $\hat{\pi}_0(\hat{t})$ is close to zero, i.e., when the filter removes many PC nulls that are unlikely to be false. Hence, \textit{filtering effectively comes at the cost of more conservative $k$-FWER control}. The proof of \thmref{AdaBon_FWER_control} is provided in \appref{AdaBon_FWER_control}; for $k = 1$, the $k$-FWER upper bound presented in \eqref{AdafilterBonControl} is tighter than the one presented in Theorem 4.2 of \cite{Wang2022}, who simply only proved that AdaFilter-Bon controls the FWER at level $\alpha$.

The thrust of this paper is to estimate $\pi_0(t)$ for $t \in [0,1]$ and incorporate it into AdaFilter-Bon’s rejection rule to bring its $k$-FWER control level closer to $\alpha$, thereby yielding a more powerful multiple testing procedure. Specifically, we replace the AdaFilter-Bon rejection threshold defined in \eqref{t_tilde} with
\begin{equation}\label{rej_thresh}
    \hat{t}_{\theta} \equiv \sup \left \{ t \in [0,1] : \hat{\pi}_0(t) \cdot t \cdot 
        \sum^m_{i = 1} I \{ F_i < t \} \leq k \alpha  \right\}
\end{equation}
instead, where $\hat{\pi}_0(t) \geq 0$ is a data-adaptive estimate of $\pi_0(t)$. Heuristically, we expect the threshold in \eqref{rej_thresh} to yield rejections with a $k$-FWER level bounded above by
\begin{equation*}
    \alpha \cdot \mathbb{E} \left[ \frac{\pi_0(\hat{t}_{\theta})}{\hat{\pi}_0(\hat{t}_{\theta})} \right], 
\end{equation*}
rather than the right-hand side of \eqref{AdafilterBonControl}. Thus, to ensure $k$-FWER control at level $\alpha$, $\hat{\pi}_0(\hat{t}_{\theta})$ should ideally satisfy $\hat{\pi}_0(\hat{t}_{\theta}) \approx \pi_0(\hat{t}_{\theta})$, or at least $\hat{\pi}_0(\hat{t}_{\theta}) \gtrsim \pi_0(\hat{t}_{\theta})$. In light of the data-\textbf{ada}ptive post-filter null proportion estimator used in forming \eqref{rej_thresh}, we refer to this new multiple testing procedure as AdaFilter-\textbf{Ada}Bon.

\section{AdaFilter-AdaBon}\label{sec:AdaFilterAdaBon}
\subsection{AdaFilter-AdaBon methodology}
Let $\theta \in (0,1)$ be a fixed tuning parameter. By the validity of Bonferroni-combined PC $p$-values, it holds that $\Pr(S_i > \theta t) \geq 1 - \theta t$ for all $t \in [0,1]$ if $H^{u/n}_i$ is true. Hence,  
\begin{align}
     \frac{\sum^m_{i = 1} I \{ F_i < t \} I \{ S_i \geq \theta t \}}{\sum^m_{i = 1} I \{ F_i < t \} } &\geq  \frac{\sum^m_{i = 1} I \{ F_i < t \} I \{ \text{$H^{u/n}_i$ is true} \} I \{ S_i \geq \theta t \}}{\sum^m_{i = 1} I \{ F_i < t \} }  \nonumber \\
     & \gtrsim  (1 - \theta t) \frac{\sum^m_{i = 1} I \{ F_i < t \} I\{ \text{$H^{u/n}_i$ is true} \}}{\sum^m_{i = 1} I \{ F_i < t \} }. \nonumber
\end{align}
Rearranging the above display yields our estimator for $\pi_0(t)$:
\begin{equation}\label{pi_hat}
    \hat{\pi}_0(t) \equiv \hat{\pi}_0(t;\theta) = \frac{\sum^m_{i = 1} I\{ F_i < t \} I\{ S_i \geq \theta t \}}{(1 - \theta t) \sum^m_{i = 1} I \{ F_i < t \} }.
\end{equation}
As the value of $\theta$ increases while $t$ remains constant, the bias of $\hat{\pi}_0(t)$ decreases and its variance increases. In our experiments, setting $\theta = 0.5$ generally strikes a good balance in this bias–variance trade-off. Having defined $\hat{\pi}_0(t)$, we now formally present the AdaFilter-AdaBon procedure below:
\begin{definition}[AdaFilter-AdaBon]\label{def:AdafilterBonpi}
    Let $\alpha \in [0,1]$ be a $k$-FWER target for a given $k \in \mathbb{N}$, $u \in \{2, \dots, m\}$ be a replicability level, $\theta \in (0,1)$ be a tuning parameter, and $S_i$ and $F_i$ be as defined in \eqref{S_i} and \eqref{F_i}, respectively. For $i=1, \dots, m$, reject $H^{u/n}_i$ if $S_i < \hat{t}_{\theta}$, where
    \begin{align}
        \hat{t}_{\theta} \equiv \sup \left \{ t \in [0,1] : t \cdot 
        \frac{\sum^m_{i = 1} I \{ F_i < t \} I\{ S_i \geq \theta t \}}{1 - \theta t} \leq k \alpha \right\}. \label{t_hat}
    \end{align}
\end{definition}

Note that \eqref{t_hat} is equivalent to the threshold previewed in \eqref{rej_thresh} of \secref{preview} when $\hat{\pi}_0(t)$ is defined as \eqref{pi_hat}. To simplify the implementation of AdaFilter-AdaBon in practice, we recommend defining
\begin{equation*}
    \mathcal{G} = \left\{ G \in \{ 0,1 \} \cup \{ F_i, S_i, S_i/\theta \}^m_{i=1} : G \leq 1 \right\}
\end{equation*}
and using the rejection threshold
\begin{equation}\label{t_tilde_theta}
    \breve{t}_{\theta} \equiv \max \left\{  t \in  \mathcal{G}: t \cdot 
        \frac{\sum^m_{i = 1} I \{ F_i < t \} I\{ S_i \geq \theta t \}}{1 - \theta t} \leq k \alpha \right\}
\end{equation}
as a surrogate for $\hat{t}_{\theta}$. The exact value of $\breve{t}_{\theta}$ can be computed in $O(m)$ time by evaluating the set-building condition in \eqref{t_tilde_theta} for each element of $\mathcal{G}$. Although $\breve{t}_{\theta} \not\equiv \hat{t}_{\theta}$, the theorem below states that $\breve{t}_{\theta}$ and $\hat{t}_{\theta}$ yield the same rejection set when thresholding the PC $p$-values in $(S_i)^m_{i=1}$. 

\begin{theorem}[$\breve{t}_{\theta}$ as a surrogate for $\hat{t}_{\theta}$]\label{thm:rej_eq}
Let $\hat{t}_{\theta}$ and $\breve{t}_{\theta}$ be as defined in \eqref{t_hat} and \eqref{t_tilde_theta}, respectively. It holds that
\begin{equation*}
    \left\{ i \in \{ 1, \dots, m \} : S_i < \hat{t}_{\theta}  \right\} = \left\{ i \in \{ 1, \dots, m \} : S_i < \breve{t}_{\theta}  \right\}.
\end{equation*}
\end{theorem}
The proof of \thmref{rej_eq}, given in \appref{rej_eq}, is based on the fact that $\sum^m_{i = 1} I \{ F_i < t \} I\{ S_i \geq \theta t \}$ and $\left\{ i \in \{ 1, \dots, m \} : S_i < t  \right\}$ are both step functions in $t$.

\subsection{Asymptotic results for AdaFilter-AdaBon}\label{sec:asymp}

In this subsection, we provide the result showing that AdaFilter-AdaBon (\defref{AdafilterBonpi}) asymptotically controls the $k$-FWER at level $\alpha$ as $m \to \infty$. The following assumption underpins our asymptotic result:
\begin{assumption}\label{assump:asymp}
    For any given $m$, let
    \begin{equation*}
        m_0 \equiv \sum^m_{i = 1} I \{ \text{$H^{u/n}_i$ is true} \} \quad \text{and} \quad m_1 \equiv \sum^m_{i = 1} I \{ \text{$H^{u/n}_i$ is false} \}
    \end{equation*}
    be the number of true and false PC nulls, respectively. The following limits exist almost surely for each $t \in (0,t]$:
    \begin{equation}\label{FS_0}
        \lim_{m \rightarrow \infty} \frac{\sum^m_{i = 1} I \{ S_i < t \} I \{ \text{$H^{u/n}_i$ is true} \}}{m_0}  = \tilde{S}_0(t), \quad \lim_{m \rightarrow \infty} \frac{\sum^m_{i = 1} I \{ F_i < t \} I \{ \text{$H^{u/n}_i$ is true} \}}{m_0}  = \tilde{F}_0(t)  
    \end{equation}
    \begin{equation}\label{FS_1}
        \lim_{m \rightarrow \infty} \frac{\sum^m_{i = 1} I \{ S_i < t \} I \{ \text{$H^{u/n}_i$ is false} \}}{m_1} = \tilde{S}_1(t), \quad  \lim_{m \rightarrow \infty} \frac{\sum^m_{i = 1} I \{ F_i < t \} I \{ \text{$H^{u/n}_i$ is false} \}}{m_1}  = \tilde{F}_1(t),
    \end{equation}    
    where $\tilde{S}_0, \tilde{F}_0,\tilde{S}_1$, and $\tilde{F}_1$ are continuous function. For any $t',t'' \in (0,1]$ where $t' \leq t''$, it holds that
    \begin{equation}\label{SF}
        0 < \tilde{S}_0(t') \leq t' \tilde{F}_0(t''). 
    \end{equation}
    Moreover, the following limit also exists:
    \begin{equation}\label{pi0_limit}
        \lim_{m \rightarrow \infty} \frac{m_0}{m} = \pi_0 \in (0,1).
    \end{equation}
\end{assumption}

The convergence properties stated in \eqref{FS_0}, \eqref{FS_1}, and \eqref{pi0_limit} are typical in the asymptotic analysis literature; for example, see \citet[Sec 2.2]{Storey2004}, \citet[Sec 4.2]{Wang2022}, and \citet[Sec 3.1]{Zhang2020}. Any type of dependence where \eqref{FS_0} and \eqref{FS_1} can hold is what we consider weak dependence. For instance, dependence in finite-dimensional blocks, autoregressive dependence, and mixing distributions are candidates for weak dependence. The properties relating to $\tilde{S}_0$ and $\tilde{F}_0$ in \eqref{SF} are not standard in the asymptotic analysis literature but arise naturally from the \textit{conditionally validity} of Bonferroni-combined PC $p$-values:
\begin{lemma}[Conditional validity]\label{lem:conditional_validity}
    Suppose $H^{u/n}_i$ is true and the $p$-values $P_{i1},\dots,P_{in}$ are all valid and mutually independent. Then for any fixed $t', t'' \in [0,1]$ where $t' \leq t''$, we have that
    \begin{equation}\label{VC}
        \Pr(S_i < t' | F_i < t'')= \frac{\Pr(S_i < t', F_i < t'')}{\Pr(F_i < t'')} = \frac{\Pr(S_i < t')}{\Pr(F_i < t'')}   \leq t'
    \end{equation}
    if $\Pr(F_i < t') > 0$. 
\end{lemma} 
Rearranging \eqref{VC} gives $\Pr(S_i < t') \leq t' \cdot \Pr(F_i < t'')$ when $H^{u/n}_i$ is true, which is analogous to the property stated in \eqref{SF} when \eqref{FS_0} and \eqref{pi0_limit} hold. The proof for \lemref{conditional_validity} is provided in \appref{conditional_validity}. 

Below, we present our asymptotic result, which states that if $k$ remains a constant proportion of $m$ as $m \to \infty$, then AdaFilter-AdaBon asymptotically controls the $k$-FWER below $\alpha$.

\begin{theorem}[AdaFilter-AdaBon $k$-FWER control]\label{thm:AdaBon_kFWER_control}
    Suppose \assumpref{asymp} holds and
    \begin{equation*}
         k \equiv k(m) = \omega \cdot m
    \end{equation*}
    for a constant $\omega \in (0,1)$. If there exists a threshold $t \in (0,1]$ such that
    \begin{equation}\label{lim_cond}
        \lim_{m \to \infty} t \cdot \frac{\sum^m_{i=1} I \{ F_i < t \} I\{ S_i \geq \theta t \} }{m (1 - \theta t)} < \omega \alpha
    \end{equation}
    with probability $1$, then AdaFilter-AdaBon (\defref{AdafilterBonpi}) has the following asymptotic $k$-FWER property:
    \begin{equation*}
        \limsup_{m \to \infty} \left(k\text{-FWER}(\mathcal{R})\right) \leq \alpha,
    \end{equation*}
    where $\mathcal{R} = \{ i \in \{ 1, \dots, m \} : S_i \leq \hat{t}_{\theta} \}$.
\end{theorem}
The technical condition stated in \eqref{lim_cond} ensures that $\hat{t}_{\theta}$ is well behaved in the limit. The proof of \thmref{AdaBon_kFWER_control}, given in \appref{AdaBon_kFWER_control}, first establishes that the expected number of false discoveries is asymptotically bounded by $k\alpha$, and then applies Markov’s inequality to show that the $k$-FWER is below $\alpha$.

\section{Simulation studies}\label{sec:num}
\subsection{Simulated data}\label{sec:sim_data}
Consider a meta-analysis consisting of $m = 500$ features across $n = 4$ studies. For each feature $i$ and study $j$, let $\mu_{ij} \in \mathbb{R}$ be an effect parameter, and consider the null hypothesis $H_{ij}: \mu_{ij} = 0$ with the corresponding alternative $\mu_{ij} \neq 0$. The data for this meta-analysis are simulated according to the steps provided below.

For a given correlation parameter $\rho \in [-1, 1]$, we arbitrarily partition the $500$ features into $b$ blocks of size $500/b$ each, where 
\begin{equation*}
   b = \begin{cases}
5 & \text{if } \rho \geq 0, \\
250 & \text{if } \rho < 0.
\end{cases}
\end{equation*}
We then generate standard normal variables $\epsilon_{ij} \sim \text{N}(0,1)$ for each pair $(i,j)$, subject to the following dependence structure:
\begin{itemize}
    \item The sequences $(\epsilon_{i1})_{i=1}^{500}$, $(\epsilon_{i2})_{i=1}^{500}$, $(\epsilon_{i3})_{i=1}^{500}$, and $(\epsilon_{i4})_{i=1}^{500}$ are mutually independent;
    \item Within each study $j$, 
    \begin{equation*}
        \text{cor}(\epsilon_{ij},  \epsilon_{i'j}) = 
        \begin{cases}
            \rho, &\text{if $i$ and $i'$ belong to the same block and $i \neq i'$;}\\
            1, &\text{if $i = i'$;}\\
            0, &\text{if $i$ and $i'$ do not belong to the same block}.
        \end{cases}
    \end{equation*}
\end{itemize}
The following process is then used to generate the $p$-values $(P_{ij})_{500 \times 4}$. For each $i = 1, \dots, 500$:
\begin{enumerate}[\indent (1)]
    \item Generate $c_i \sim \text{Bernoulli}(\pi_1)$, where $\pi_1 \in [0,1]$ is a given signal density parameter.
        \begin{enumerate}[(i)]
            \item If $c_i = 0$, set $(\mu_{i1},\mu_{i2},\mu_{i3},\mu_{i4}) = (0,0,0,0)$;
            \item If $c_i = 1$, sample $(\mu_{i1},\mu_{i2},\mu_{i3},\mu_{i4})$ uniformly from the set $\{ 0, 4 \}^4$.
        \end{enumerate}
    \item For $j = 1, \dots, 4$, compute $P_{ij} = 1 - \Phi(\mu_{ij} + \epsilon_{ij})$.
\end{enumerate}
Hence, by construction, the $p$-values in $(P_{ij})_{500 \times 4}$ are valid and independent across studies, while exhibiting block-wise equicorrelation within each study.

\subsection{Compared methods}
The following multiple testing procedures, representing a mix of state-of-the-art and well-established approaches, are considered in our simulations:
\begin{enumerate}[\indent $(a)$]
    \item AdaFilter-AdaBon: \defref{AdafilterBonpi} with $\theta = 0.5$.
    \item AdaFilter-Bon: \defref{AdafilterBon}.
    \item Bonferroni: \defref{Bon}.
    \item Hochberg: The step-up $k$-FWER procedure described in Theorem 4.2 of \citet{Sarkar2007}.
    \item Adaptive Bonferroni: Definition 1 of \citet{Guo2009} with tuning parameter $\lambda = 0.5$.
    \item Adaptive Hochberg: Definition 4 of \citet{Sarkar2012} with $\hat{n}(\kappa)$ where $\kappa = 490$.
\end{enumerate}
Let $W_{u}(\cdot)$ denote the cdf of a chi-squared distribution with $2 \cdot (n - u + 1)$ degrees of freedom. Methods $(c)$---$(f)$ are implemented in our simulations with the PC $p$-value inputs $(P^{u/n}_{i})_{i=1}^m$, where each $P^{u/n}_i$ is constructed as the Fisher-combined PC $p$-value \citep[Eq.~5]{Benjamini2008}:
\begin{equation*}
    f_{\text{Fisher}}(P_{i1},\dots,P_{in}; u) \equiv 1 - W_{u} \left(-2 \sum^{n}_{j = u} \log P_{i(j)} \right).
\end{equation*}
This choice is motivated by the fact that Fisher-combined PC $p$-values are less conservative than their Bonferroni-combined PC $p$-values \citep[Sec. 2.1]{Bogomolov2023}. By contrast, methods $(a)$ and $(b)$, by definition, must operate on Bonferroni-combined PC $p$-values.

\subsection{Simulation results}\label{sec:sim_res}
We apply methods $(a)$--–$(f)$ to the simulated data described in \secref{sim_data} to test the PC nulls $(H^{u/4}_i)_{i=1}^{500}$ for $u \in \{2, 3, 4 \}$, targeting FWER control at level $\alpha = 0.05$. The following parameter ranges are considered for $\pi_1$ and $\rho$:
\begin{itemize}
    \item $\pi_1 \in \{ 0.025, 0.05, 0.075, 0.10, 0.125, 0.15 \}$.
    \item $\rho \in \{ -0.8, -0.2, 0.2, 0.8 \}$.
\end{itemize}
By design, each $H^{u/n}_i$ becomes more likely to be false as $\pi_1$ increases from $0.025$ to $0.15$. Moreover, the block-wise correlation strength among the $p$-values in $(P_{ij})_{500 \times 4}$ shifts from strongly negative to strongly positive as $\rho$ increases from $-0.8$ to $0.8$. The simulation results are presented in \figref{sim_k1}, where the empirical FWER and TPR level of the compared methods are computed based on $1,000$ repetitions. The following observations can be made:

\begin{itemize}
    \item AdaFilter-AdaBon demonstrated finite-sample FWER control across all settings of $\pi_1$, $u$, and $\rho$. This is reassuring, given that we only established its theoretical control guarantee asymptotically as $m \to \infty$ (\thmref{AdaBon_kFWER_control}). We also observe that AdaFilter-AdaBon’s FWER level is uniformly greater than that of AdaFilter-Bon, indicating that the former is less conservative than the latter. Outside of AdaFilter-AdaBon and AdaFilter-Bon, the remaining methods generally exhibited extremely conservative FWER control. As $\rho$ increases from $-0.8$ to $0.8$, the FWER level for all methods decreases slightly.
    \item AdaFilter-AdaBon outperforms AdaFilter-Bon under most parameter settings, with notably higher power when $\pi_1 \geq 0.10$ and $u = 2$. Only when $\pi_1 = 0.025$ and $u = 4$ does AdaFilter-Bon match the power of AdaFilter-AdaBon. The remaining methods exhibited substantially lower power than both AdaFilter-AdaBon and AdaFilter-Bon, particularly when $u = 4$, where their power was close to zero.
\end{itemize}

Hence, our simulations demonstrate that AdaFilter-AdaBon delivers state-of-the-art performance, especially when the replicability requirement ($u$) is small and the signal density ($\pi_1$) is large. Even in less favorable settings, its power remains comparable to that of its predecessor, AdaFilter-Bon. 

Additional simulation results comparing the methods under $k$-FWER control for $k = 5$ and $10$, using the same $\pi_1$ and $\rho$ settings as above, are provided in \appref{additional_sim_results}. The simulation results presented in this section and in \appref{additional_sim_results} can be reproduced in \texttt{R} via the steps provided in \href{https://github.com/ninhtran02/AdaFilterAdaBon}{https://github.com/ninhtran02/AdaFilterAdaBon}.

\begin{figure}
    \centering
    \includegraphics[width=0.85\linewidth]{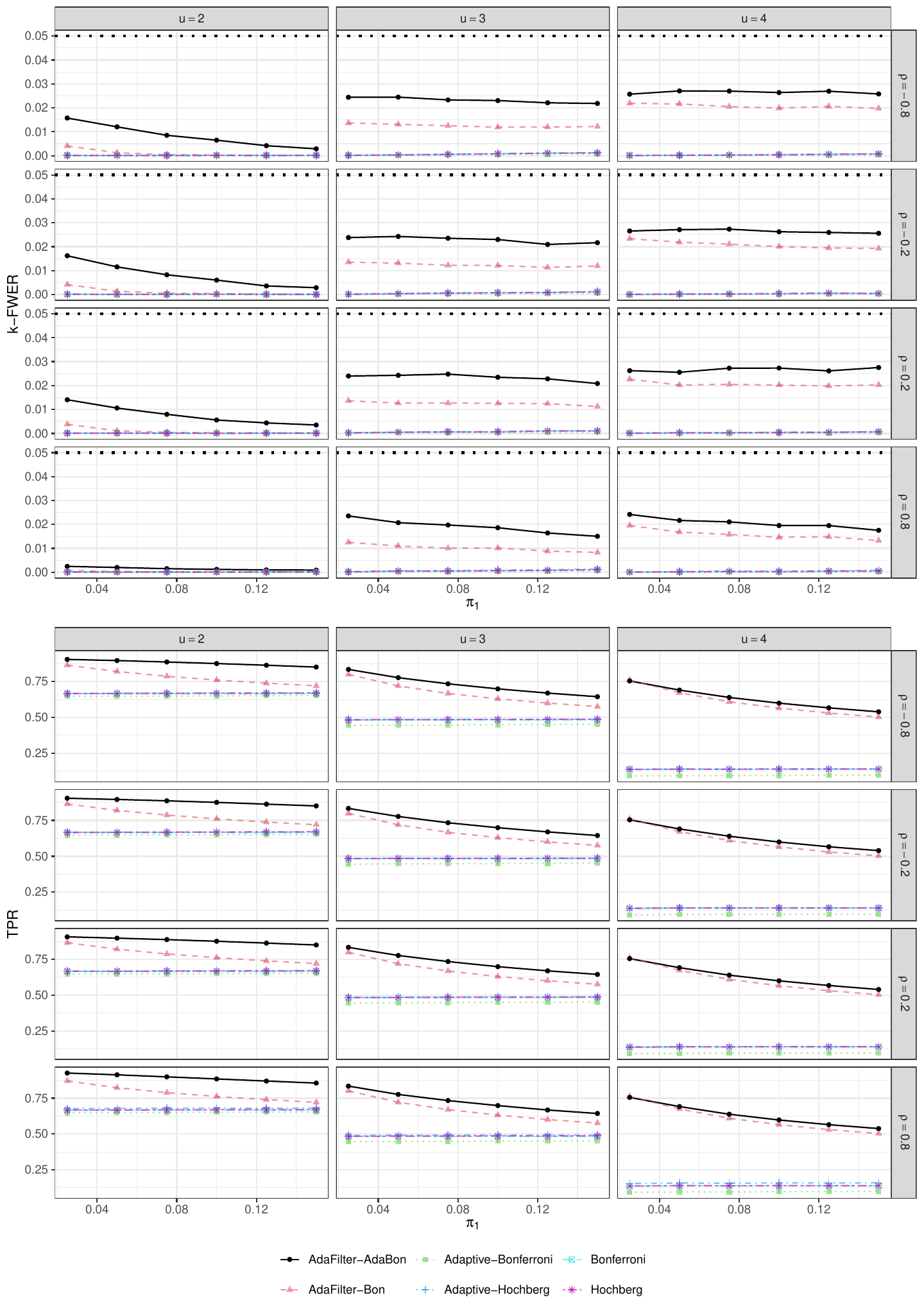}
    \caption{Empirical FWER and TPR levels of the compared methods for the parameter settings $\alpha = 0.05$, $u \in \{ 2,3,4 \}$, $\pi_1 \in \{ 0.025, 0.05, 0.075, 0.10, 0.125, 0.15 \}$, and $\rho \in \{ -0.8, -0.2, 0.2, 0.8 \}$.}
    \label{fig:sim_k1}
\end{figure}

\section{Discussion}\label{sec:discussion}
We proposed AdaFilter-AdaBon, a $p$-value-based procedure for testing partial conjunction hypotheses with asymptotic $k$-FWER control. By combining filtering with a post-filter null proportion estimator, our method improves power over its predecessor, AdaFilter-Bon \citep{Wang2022}, while maintaining robust finite-sample $k$-FWER control, as shown in our simulation studies.

AdaFilter-AdaBon can be extended to control the false exceedance rate (FDX) and the false discovery rate (FDR). For a set $\mathcal{R} \subseteq \{ 1 , \dots, m \}$ of rejected PC nulls and tolerance level $\gamma \in (0,1)$, the FDX is defined as
\begin{equation*}
    \text{FDX} \equiv \text{FDX}(\mathcal{R}; \gamma) = \Pr\left( \frac{\sum^m_{i=1} I \{ i \in \mathcal{R} \} I \{ \text{$H^{u/n}_i$ is true} \} }{1 \vee \sum^m_{i = 1} I \{ i \in \mathcal{R} \} } \geq \gamma \right),
\end{equation*}
which is the probability that the proportion of false discoveries among the total number of discoveries exceeds $\gamma$. The false discovery rate is defined as
\begin{equation*}
    \text{FDR} \equiv \text{FDR}(\mathcal{R}) = \Pr\left[ \frac{\sum^m_{i=1} I \{ i \in \mathcal{R} \} I \{ \text{$H^{u/n}_i$ is true} \} }{1 \vee \sum^m_{i = 1} I \{ i \in \mathcal{R} \} }  \right] \in (0,1),
\end{equation*}
which is the expected proportion of false discoveries among the total number of discoveries. The following procedure, based on AdaFilter-AdaBon, controls both FDX and FDR asymptotically: 

\begin{definition}\label{def:AugmentedAdafilterBonpi}
Let $\alpha \in [0,1]$ be a $k$-FWER target for a given $k \in \mathbb{N}$, $u \in \{2, \dots, m\}$ be a replicability level, $\theta \in (0,1)$ be a tuning parameter, and $F_i$ and $S_i$ be as defined in \eqref{F_i} and \eqref{S_i}, respectively. For $i = 1, \dots, m$, reject $H^{u/n}_i$ if $S_i \leq \hat{\tau}$, where
\begin{align}\label{hat_tau}
    \hat{\tau} \equiv \sup \left\{ \tau \in [0,1]:  \frac{ \sum^{m}_{i = 1} I \{ \hat{t}_{\theta}  \leq S_i \leq \tau \} + k}{1 \vee \sum^{m}_{i = 1} I \{ S_i \leq \tau \}} \leq \gamma \right\}
\end{align}
and $\hat{t}_{\theta}$ is AdaFilter-AdaBon's rejection threshold as defined in \eqref{t_hat}.
\end{definition}

\begin{corollary}[FDX and FDR control]\label{cor:AdaBon_FDX_control}
    If \thmref{AdaBon_kFWER_control} holds, then the procedure in \defref{AugmentedAdafilterBonpi} has the following asymptotic error properties:
    \begin{equation*}
        \limsup_{m \to \infty} \text{FDX}(\mathcal{R}) \leq \alpha \quad \text{and} \quad \limsup_{m \to \infty} \text{FDR}(\mathcal{R}) \leq \alpha + \gamma,
    \end{equation*}
    where $\mathcal{R} = \{ i \in \{ 1, \dots, m \} : S_i \leq \hat{\tau} \}$.
\end{corollary}
The proof of \corref{AdaBon_FDX_control} is given in \appref{AdaBon_FDX_control}. Within the proof, it is shown that the procedure in \defref{AugmentedAdafilterBonpi} controls FDX and FDR indirectly through $k$-FWER control. Hence, an avenue for future work is to develop more direct approaches for controlling FDX and FDR that incorporate filtering and our post-filter null proportion estimator. This may lead to improved power over the procedure presented in \defref{AugmentedAdafilterBonpi}.

\bibliographystyle{apalike}
\bibliography{References}

\appendix
\section{Results and proofs relating to \secref{problem}}

\subsection{Technical lemmas for proving \thmref{AdaBon_FWER_control}}\label{app:threshold}
For all $i \in \{ 1, \dots, m \}$, let
\begin{equation}\label{t_for_proof}
    \hat{t}_{i}  = \sup \left\{ t \in [0,1] : t \cdot \left( 1 + \sum_{\ell \neq i} I \{ F_{\ell} < t \} \right) \leq \alpha \right\}
\end{equation}
be a threshold similar to $\hat{t}$ in \eqref{t_tilde}, but computed without data from feature $i$. We now state and prove two key lemmas involving $\hat{t}_i$, which are required for the proof of \thmref{AdaBon_FWER_control} in \appref{AdaBon_FWER_control}.

\begin{lemma}[$\hat{t}_{i}  \leq \hat{t}$]\label{lem:t_hat_leq}
    Let $\hat{t}$ and $\hat{t}_{i}$ be as defined in \eqref{t_tilde} and \eqref{t_for_proof}, respectively. Then $\hat{t}_{i}  \leq \hat{t}$.
\end{lemma}
\begin{proof}[Proof of \lemref{t_hat_leq}]
    Note that
    \begin{align*}
        t \cdot \sum^m_{\ell = 1} I \{ F_{\ell} < t \} =  t \cdot \left( I\{ F_i < t \} + \sum_{\ell \neq i} I \{ F_{\ell} < t \} \right) \leq t \cdot \left(1 + \sum_{\ell \neq i} I \{ F_{\ell} < t \} \right).
    \end{align*}
    As a consequence of the display above, we have that
    \begin{align*}
        \hat{t} &\equiv \sup \left\{ t \in [0,1] : t \cdot \sum^m_{\ell = 1} I \{ F_{\ell} < t \} \leq \alpha  \right\} \\
        &\geq \sup \left\{ t \in [0,1] : t \cdot \left(1 + \sum_{\ell \neq i} I \{ F_{\ell} < t \} \right) \leq \alpha  \right\} \equiv \hat{t}_{i}.
    \end{align*}
\end{proof}

\begin{lemma}[$\hat{t}_{i}  = \hat{t}$ if $F_i < \hat{t}$]\label{lem:gamma_i_equal_t_hat}
    Let $\hat{t}$ and $\hat{t}_{i}$ be as defined in \eqref{t_tilde} and \eqref{t_for_proof}, respectively. If $F_i < \hat{t}$, then $\hat{t}_{i}  = \hat{t}$.
\end{lemma}
\begin{proof}[Proof of \lemref{gamma_i_equal_t_hat}]
    If $F_i < \hat{t}$, it must be the case that
    \begin{align}
        \hat{t} &\equiv \sup \left\{ t \in [0,1] : t \cdot \sum^m_{\ell = 1} I \{ F_{\ell} < t \} \leq \alpha \right\} \nonumber \\
        &= \sup \left\{ t \in (F_i,1] : t \cdot \sum^m_{\ell = 1} I \{ F_{\ell} < t \}  \leq \alpha \right\} \nonumber \\
        &= \sup \left\{ t \in (F_i,1] : t \cdot \left( 1 +  \sum_{\ell \neq i} I \{ F_{\ell} < t \} \right) \leq \alpha \right\} \nonumber
    \end{align}
    where
    \begin{equation}\label{some_t_star}
       \exists \ t^* \in (F_i, 1] \quad \text{such that} \quad  t^* \cdot \left( 1 +  \sum_{\ell \neq i} I \{ F_{\ell} < t^* \} \right) \leq \alpha.
    \end{equation}
    From \eqref{t_for_proof}, we have that
    \begin{align}
        \hat{t}_{i}  &\equiv \sup \left\{ t \in [0,1] : t \cdot \left( 1 +  \sum_{\ell \neq i} I \{ F_{\ell} < t \} \right) \leq \alpha \right\} \nonumber \\
        &= \sup \left\{ t \in (F_i,1] : t \cdot \left( 1 +   \sum_{\ell \neq i} I \{ F_{\ell} < t \} \right) \leq \alpha \right\} \nonumber
    \end{align}
    where the second equality is a consequence of \eqref{some_t_star}.
    Hence, we can conclude that $\hat{t} = \hat{t}_{i} $.
\end{proof}

\subsection{Proof of \thmref{AdaBon_FWER_control}}\label{app:AdaBon_FWER_control}
\begin{proof}[\unskip\nopunct]
Let $V = \sum^m_{i=1} I\{ S_i < \hat{t} \} I \{ \text{$H^{u/n}_i$ is true} \}$ denote the number of false discoveries made by AdaFilter-Bon (\defref{AdafilterBon}), and $\hat{t}_i$ be the rejection threshold defined in \eqref{t_for_proof}. We have that
\begin{align*}
        \mathbb{E}[V] &= \mathbb{E} \left[ \sum^m_{i=1} I\{ S_i < \hat{t} \} I \{ \text{$H^{u/n}_i$ is true} \} \right] \\
        &= \mathbb{E} \left[ \sum^m_{i=1} I\{ S_i < \hat{t} \} I\{ F_i < \hat{t} \} I \{ \text{$H^{u/n}_i$ is true} \} \right] \\
        &= \mathbb{E} \left[ \sum^m_{i=1} I\{ S_i < \hat{t}_{i}  \} I\{ F_i < \hat{t}_{i}  \}  I \{ \text{$H^{u/n}_i$ is true} \} \right] 
\end{align*}
where the second equality is a consequence of $F_i \leq S_i$, and the third equality is a result of \lemref{gamma_i_equal_t_hat}. Continuing on from the previous display, we have that

\begin{align*}
        \mathbb{E}[V] &= \sum^m_{i=1}  \mathbb{E} \left[ \mathbb{E} \left[ I\left\{ S_i < \hat{t}_{i}  \right\} \Big| \hat{t}_{i} , F_i < \hat{t}_{i}   \right] I\{ F_i < \hat{t}_{i}  \}  I \{ \text{$H^{u/n}_i$ is true} \} \right] \\
         &\leq \sum^m_{i=1}  \mathbb{E} \left[ \hat{t}_{i}  \cdot I\{ F_i < \hat{t}_{i}  \}  I \{ \text{$H^{u/n}_i$ is true} \} \right] \\
        &\leq \sum^m_{i=1}  \mathbb{E} \left[ k \alpha \cdot   \frac{1}{1 + \sum_{\ell \neq i} I\{ F_{\ell} \leq \hat{t}_{i}  \} } \cdot I\{ F_i \leq \hat{t}_{i}  \}  I \{ \text{$H^{u/n}_i$ is true} \} \right] \\
      &\leq \sum^m_{i=1}  \mathbb{E} \left[  \frac{k \alpha}{\sum^m_{\ell = 1} I\{  F_{\ell} \leq \hat{t}_{i}  \} } I\{ F_i \leq \hat{t}_{i}  \}  I \{ \text{$H^{u/n}_i$ is true} \}    \right] 
    \end{align*}
    where the first inequality is a consequence of \lemref{conditional_validity} and the independence between $S_i$ and $\hat{t}_i$, and the second inequality follows from \eqref{t_for_proof}. Continuing on from the display above, we have that 
    \begin{align*}
        \mathbb{E}[V] &\leq \sum^m_{i=1}  \mathbb{E} \left[  \frac{k \alpha}{\sum^m_{\ell = 1} I\{  F_{\ell} \leq \hat{t}_{i}  \} } I\{ F_i \leq \hat{t}  \}  I \{ \text{$H^{u/n}_i$ is true} \}    \right]  \\
        &= \sum^m_{i=1}  \mathbb{E} \left[  \frac{k \alpha}{\sum^m_{\ell = 1} I\{  F_{\ell} \leq \hat{t}  \}} I\{ F_i \leq \hat{t}  \}  I \{ \text{$H^{u/n}_i$ is true} \}    \right] \\
      &= k \alpha \cdot   \mathbb{E} \left[  \frac{\sum^m_{i=1} I\{ F_i \leq \hat{t}  \}  I \{ \text{$H^{u/n}_i$ is true} \} }{\sum^m_{\ell = 1} I\{ F_{\ell} \leq \hat{t}  \} }     \right] 
    \end{align*}
    where the inequality is a result of \lemref{t_hat_leq} and the first equality is a consequence of \lemref{gamma_i_equal_t_hat}. By Markov’s inequality, it follows that
    \begin{align*}
        k\text{-FWER}(\mathcal{R}) = \Pr(V \geq k)  \leq \alpha  \mathbb{E} \left[  \frac{\sum^m_{i=1} I\{ F_i \leq \hat{t}  \}  I \{ \text{$H^{u/n}_i$ is true} \} }{\sum^m_{\ell = 1} I\{ F_{\ell} \leq \hat{t}  \} }     \right] .
    \end{align*}
\end{proof}

\section{Results and proofs relating to \secref{AdaFilterAdaBon}}
\subsection{Proof of \lemref{conditional_validity}}\label{app:conditional_validity}
\begin{proof}[\unskip\nopunct]
We have that
    \begin{align*}
        \Pr(S_{i} < t' | F_{i} < t'' ) &= \frac{\Pr(S_{i} < t' , F_{i} < t'' )}{\Pr(F_{i} < t'' )} \\
        & \leq \frac{\Pr(S_{i} < t'  )}{\Pr(F_{i} < t'' )} \\
        & \leq \frac{\Pr(S_{i} < t'  )}{\Pr(F_{i} < t' )} \\
        & = \frac{\Pr(S_{i} < t', F_{i} < t'  )}{\Pr(F_{i} < t' )} \\
        & = \Pr(S_i < t' | F_i < t')\\
        & \leq t'
    \end{align*}
    where the second equality is a result of $F_i \leq S_i$, and the last inequality is a consequence of Lemma 4.1 from \cite{Wang2022}.
\end{proof}

\subsection{Proof of \thmref{rej_eq}}\label{app:rej_eq}
\begin{proof}[\unskip\nopunct]

Without loss of generality, suppose $\mathcal{G}$ is ordered and is of size $|\mathcal{G}| = M$. To denote the ordered elements of $\mathcal{G}$, we let
\begin{equation*}
    \mathcal{G} \equiv \left\{ G_1, G_2, \dots, G_{M} \right\}
\end{equation*}
where $0 = G_1 \leq G_2 \leq \cdots \leq G_{M} = 1$. Consider the functions $f:[0,1] \longrightarrow [0, \infty)$ and $g:[0,1] \longrightarrow 2^{\{ 1, \dots, m \}}$ defined below:
\begin{equation*}
    f(t) \equiv \sum^m_{i = 1} I \{ F_i < t \} I\{ S_i \geq \theta t \}  \quad \text{and} \quad g(t) \equiv \{ i \in \{ 1, \dots, m \} : S_i < t \}. 
\end{equation*}
It is not difficult to see that both $f(t)$ and $g(t)$ are step functions on the following intervals for $t$:
\begin{equation*}
    \mathcal{I}_1 = [G_1,G_2), \quad \mathcal{I}_2 = [G_2, G_3), \quad \dots , \quad \mathcal{I}_{M - 1} = [G_{M-1}, G_{M}), \quad \mathcal{I}_{M} = \{ G_{M} \},
\end{equation*}
where $\bigcup^{M}_{\ell = 1} \mathcal{I}_{\ell} = [0,1]$. That is, for any $t \in [0,1]$, $f(t)$ and $g(t)$ can be written as
\begin{equation*}
    f(t) = \sum^m_{i = 1} I \{ F_i < G_{\ell} \} I\{ S_i \geq \theta \cdot G_{\ell} \} \quad  \text{ and } \quad g(t) = \{ i \in \{ 1, \dots, m \} : S_i <  G_{\ell} \} \quad 
\end{equation*}
respectively, where
\begin{equation*}
    \text{$t \in \mathcal{I}_{\ell} =
    \begin{cases}
        [ G_{\ell}, G_{\ell + 1} ), &\text{ $\ell \neq M$} \\
        \{ G_M \}, &\text{ $\ell = M$} 
    \end{cases}
    $}.
\end{equation*}

Let $\hat{\ell}_{\theta} \in \{ 1, \dots, M \}$ be the index such that $\hat{t}_{\theta} \in \mathcal{I}_{\hat{\ell}_{\theta}}$. It follows that 
\begin{equation}\label{eq_res}
    \sum^m_{i = 1} I \{ F_i < \hat{t}_{\theta} \} I\{ S_i \geq \theta \cdot \hat{t}_{\theta} \} = f(\hat{t}_{\theta}) =  \sum^m_{i = 1} I \{ F_i < G_{\hat{\ell}_{\theta}} \} I\{ S_i \geq \theta \cdot G_{\hat{\ell}_{\theta}} \} 
\end{equation}
and
\begin{equation}\label{rej_eq}
    \left\{ i \in \{ 1, \dots, m \} : S_i < \hat{t}_{\theta}  \right\} = g(\hat{t}_{\theta}) = \left\{ i \in \{ 1, \dots, m \} : S_i <  G_{\hat{\ell}_{\theta}} \right\}.
\end{equation}

Since $\mathcal{G} \subseteq [0,1]$, it follows from the definition of $\hat{t}_{\theta}$ and $\breve{t}_{\theta}$ in  \eqref{t_hat} and \eqref{t_tilde_theta}, respectively, that $\breve{t}_{\theta} \leq \hat{t}_{\theta}$. Since $G_{\hat{\ell}_{\theta}}$ is the element in $\mathcal{G}$ that is closest to $\hat{t}_{\theta}$ without being strictly greater than $\hat{t}_{\theta}$, it must hold from the inequality $\breve{t}_{\theta} \leq \hat{t}_{\theta}$ that $\breve{t}_{\theta} \leq G_{\hat{\ell}_{\theta}}$. Now consider the following result:
\begin{align*}
    G_{\hat{\ell}_{\theta}} \cdot \frac{\sum^{m}_{i=1}I\{ F_i < G_{\hat{\ell}_{\theta}} \} I\{ S_i \geq \theta \cdot G_{\hat{\ell}_{\theta}} \} }{1 - \theta \cdot G_{\hat{\ell}_{\theta}} } &= \frac{G_{\hat{\ell}_{\theta}}}{1 - \theta \cdot G_{\hat{\ell}_{\theta}} } \cdot \sum^{m}_{i=1}I\{ F_i < \hat{t}_{\theta} \} I\{ S_i \geq \theta \cdot \hat{t}_{\theta} \}  \\
    &\leq \frac{\hat{t}_{\theta}}{1 - \theta \cdot \hat{t}_{\theta}  }  \cdot \sum^{m}_{i=1}I\{ F_i < \hat{t}_{\theta} \} I\{ S_i \geq \theta \cdot \hat{t}_{\theta} \}  \\
    &\leq k \alpha
\end{align*}
where the equality is a result of \eqref{eq_res}, the first inequality is a result of $\hat{t}_{\theta} \geq G_{\hat{\ell}_{\theta}}$, and the last inequality is a result of the definition of $\hat{t}_{\theta}$ in \eqref{t_hat}. The display directly above implies that $\breve{t}_{\theta} \geq G_{\hat{\ell}_{\theta}}$ by the definition of $\breve{t}_{\theta}$ in \eqref{t_tilde_theta}. However, since we have already established that $\breve{t}_{\theta} \leq G_{\hat{\ell}_{\theta}}$, it therefore must be that $\breve{t}_{\theta} = G_{\hat{\ell}_{\theta}}$. Hence, it follows from \eqref{rej_eq} that
\begin{equation*}
    \left\{ i \in \{ 1, \dots, m \} : S_i < \hat{t}_{\theta}  \right\} = \left\{ i \in \{ 1, \dots, m \} : S_i <  G_{\hat{\ell}_{\theta}} \right\} = \left\{ i \in \{ 1, \dots, m \} : S_i <  \breve{t}_{\theta} \right\}.
\end{equation*}
\end{proof}

\subsection{Technical lemmas for proving \thmref{AdaBon_kFWER_control}}
We provide in this section technical results analogous to those of Section 5 in \cite{Storey2004}. They serve as essential ingredients in proving \thmref{AdaBon_kFWER_control}, presented in \appref{AdaBon_kFWER_control}.

\begin{lemma}\label{lem:piasymp}
    Under \assumpref{asymp}, we have for fixed constants $\theta \in (0,1)$ and $\delta \in (0,1]$ that
    \begin{equation*}
        \lim_{m \to \infty} \inf_{t \geq \delta} \left( t \cdot  \frac{\sum^m_{i = 1} I \{ F_i < t \} I \{ S_i \geq \theta t \} }{(1 - \theta t) \sum^m_{i = 1} I \{ F_i < t\} } - \frac{\pi_0 \tilde{S}_0(t) m}{\sum^m_{i = 1} I\{ F_i < t \} } \right) \geq 0
    \end{equation*}
with probability $1$.
\end{lemma}
\begin{proof}[Proof of \lemref{piasymp}]
Note that $I \{ F_i < t \} = 1$ whenever $I \{ S_i < \theta t \} = 1$, since $F_i \leq S_i$ and $\theta t < t$. This gives us the following result:
\begin{align}
    \frac{\sum^m_{i = 1} I \{ F_i < t \} I \{ S_i \geq \theta t \} }{(1 - \theta t) \sum^m_{i = 1} I \{ F_i < t\} } &= \frac{\sum^m_{i = 1} I \{ F_i < t \} (1 - I \{ S_i < \theta t \}) }{(1 - \theta t) \sum^m_{i = 1} I \{ F_i < t\} } \nonumber \\
    &= \frac{\sum^m_{i = 1} I \{ F_i < t \}  - \sum^m_{i = 1} I \{ F_i < t \} I \{ S_i < \theta t \} }{(1 - \theta t) \sum^m_{i = 1} I \{ F_i < t\} } \nonumber \\
    &= \frac{\sum^m_{i = 1} I \{ F_i < t \}  - \sum^m_{i = 1} I \{ S_i < \theta t \} }{(1 - \theta t) \sum^m_{i = 1} I \{ F_i < t\} } \nonumber \\
    &= \frac{(1/m) \sum^m_{i = 1} I \{ F_i < t \}  - (1/m) \sum^m_{i = 1} I \{ S_i < \theta t \} }{(1 - \theta t) (1/m) \sum^m_{i = 1} I \{ F_i < t\} }. \label{pisumaymp}
\end{align}
Since $F_i \leq S_i$ and $\tilde{F}_1$ is a non-decreasing function by construction, it follows from \eqref{FS_1} in \assumpref{asymp} that
\begin{equation}\label{fracprop2}
    \tilde{S}_1(\theta t) \leq \tilde{F}_1(\theta t) \leq \tilde{F}_1(t)
\end{equation}
for all $t \in (0,1]$. We also have by \eqref{SF} in \assumpref{asymp} that
\begin{equation}\label{fracprop1}
   \tilde{S}_0(\theta t) \leq \theta t \tilde{F}_0(\theta t) \leq \theta t \tilde{F}_0(t) 
\end{equation}
for all $t \in (0,1]$. 

For the next part of this proof, we adopt the following notation for convenience:
\begin{equation}\label{full_FS}
    \tilde{F}(t) \equiv \pi_0 \tilde{F}_0(t) + \pi_1 \tilde{F}_1 (t) \quad \text{and} \quad \tilde{S}(t) \equiv \pi_0 \tilde{S}_0(t) + \pi_1 \tilde{S}_1 (t).
\end{equation}
By \eqref{FS_0}, \eqref{FS_1}, and \eqref{pi0_limit} in \assumpref{asymp}, and by \eqref{pisumaymp}, we have for all $t \in (0,1]$ that
\begin{align}
    \lim_{m \to \infty}   \frac{\sum^m_{i = 1} I \{ F_i < t \} I \{ S_i \geq \theta t \} }{(1 - \theta t) \sum^m_{i = 1} I \{ F_i < t\} } &= \frac{\tilde{F}(t) - \tilde{S}(\theta t) }{(1-\theta t) \tilde{F}(t)  } \nonumber \\
    &= \frac{\pi_0 \tilde{F}_0(t) + \pi_1 \tilde{F}_1(t)  - \pi_0 \tilde{S}_0(\theta t) - \pi_1 \tilde{S}_1(\theta t)  }{(1- \theta t) \tilde{F}(t)} \nonumber \\
    &\geq \frac{\pi_0 \tilde{F}_0(t) + \pi_1 \tilde{F}_1(t)  - \pi_0 \theta t \tilde{F}_0(t) - \pi_1 \tilde{F}_1(t)  }{(1- \theta t) \tilde{F}(t)} \nonumber  \\
    &= \frac{\pi_0 (1 - \theta t) \tilde{F}_0(t)  }{(1- \theta t) \tilde{F}(t)}  \nonumber \\
    &= \frac{\pi_0  \tilde{F}_0(t)  }{\tilde{F}(t)},
    \label{inter0}
\end{align}
where the inequality is a consequence of \eqref{fracprop1} and \eqref{fracprop2}. Since $\tilde{S}_0(t) \leq t \tilde{F}_0(t)$ holds by \eqref{SF} in \assumpref{asymp}, it follows from \eqref{inter0} that
\begin{equation}\label{inter1}
    \lim_{m \to \infty} \left( t \cdot \frac{\sum^m_{i = 1} I \{ F_i < t \} I \{ S_i \geq \theta t \} }{(1 - \theta t) \sum^m_{i = 1} I \{ F_i < t\} }  - \frac{\pi_0 \tilde{S}_0(t)}{\tilde{F}(t)} \right) \geq 0.
\end{equation}
By \eqref{FS_0}, \eqref{FS_1}, and \eqref{pi0_limit} in \assumpref{asymp}, we also have that
\begin{equation}\label{inter2}
    \lim_{m \to \infty} \frac{\pi_0 \tilde{S}_0(t) m}{\sum^m_{i = 1} I \{ F_i < t \} } = \frac{\pi_0 \tilde{S}_0(t)}{\tilde{F}(t)}.
\end{equation}
By combining \eqref{inter1} and \eqref{inter2}, it holds that
\begin{equation*}
        \lim_{m \to \infty} \inf_{t \geq \delta} \left(  t \cdot \frac{\sum^m_{i = 1} I \{ F_i < t \} I \{ S_i \geq \theta t \} }{(1 - \theta t) \sum^m_{i = 1} I \{ F_i < t\} }  - \frac{\pi_0 \tilde{S}_0(t) m}{\sum^m_{i = 1} I\{ F_i < t \} } \right) \geq 0
\end{equation*}
\end{proof}

\begin{lemma}\label{lem:piasymp2}
    Under \assumpref{asymp}, we have for fixed constants $\theta \in (0,1)$ and $\delta \in (0,1]$ that
    \begin{equation*}
        \lim_{m \to \infty} \sup_{t \geq \delta} \left| \frac{\sum^m_{i = 1} I\{ S_i < t \} I \{ \text{$H^{u/n}_i$ is true} \}}{\sum^m_{i=1} I \{ F_i < t \} } - \frac{\pi_0 m  \tilde{S}_0(t)}{\sum^m_{i=1} I \{ F_i < t \} }   \right| = 0
\end{equation*}
with probability $1$.
\end{lemma}

\begin{proof}[Proof of \lemref{piasymp2}]
By the Glivenko-Cantelli theorem, it follows from \assumpref{asymp} that 
\begin{equation}\label{GC}
    \lim_{m \to \infty} \sup_{0 \leq t \leq 1} \left| \frac{\sum^m_{i = 1} I\{ S_i < t \} I \{ \text{$H^{u/n}_i$ is true} \} }{m} - \pi_0 \tilde{S}_0(t) \right| = 0 
\end{equation}
with probability 1. To prove the lemma, note that
\begin{align*}
    \lim_{m \to \infty} &\sup_{t \geq \delta} \left| \frac{\sum^m_{i = 1} I\{ S_i < t \} I \{ \text{$H^{u/n}_i$ is true} \}}{\sum^m_{i=1} I \{ F_i < t \} } - \frac{\pi_0 m  \tilde{S}_0(t)}{\sum^m_{i=1} I \{ F_i < t \} }   \right| \\
    &\leq \lim_{m \to \infty} \sup_{t \geq \delta} \left| \frac{m}{\sum^{m}_{i=1} I\{ F_i < t \} } \right| \left| \frac{\sum^m_{i = 1} I\{ S_i < t \} I \{ \text{$H^{u/n}_i$ is true} \}}{ m } - \pi_0 \tilde{S}_0(t)  \right| \\
    &\leq  \lim_{m \to \infty} \left| \frac{m}{\sum^m_{i = 1} I \{ F_i < \delta \} } \right| \sup_{t \geq \delta} \left| \frac{\sum^m_{i = 1} I\{ S_i < t \} I \{ \text{$H^{u/n}_i$ is true} \} }{m} - \pi_0 \tilde{S}_0(t) \right| = 0
\end{align*}
where the equality is a consequence of \eqref{GC}.
\end{proof}

\begin{lemma}\label{lem:piasymp3}
    Under \assumpref{asymp}, we have for fixed constants $\theta \in (0,1)$ and $\delta \in (0,1]$ that
    \begin{equation*}
        \lim_{m \to \infty} \inf_{t \geq \delta} \left( \frac{\sum^m_{i = 1} I\{ F_i < t \} I \{ S_i \geq \theta t \}}{m (1 - \theta t)} \cdot t - \frac{\sum^m_{i = 1} I\{ S_i < t \} I \{ \text{$H^{u/n}_i$ is true} \}}{m} \right) \geq 0
\end{equation*}
with probability $1$.
\end{lemma}
\begin{proof}[Proof of \lemref{piasymp3}]
We have that
\begin{align*}
    \lim_{m \to \infty} \inf_{t \geq \delta} \left( t \cdot \frac{\sum^m_{i = 1} I\{ F_i < t \} I \{ S_i \geq \theta t \}}{m (1 - \theta t)}  - \frac{\sum^m_{i = 1} I\{ S_i < t \} I \{ \text{$H^{u/n}_i$ is true} \}}{m} \right) \\
    = \lim_{m \to \infty} \inf_{t \geq \delta} \left( \frac{\sum^m_{i = 1} I\{ F_i < t \}}{m} \right) \left(   t \cdot  \frac{\sum^m_{i = 1} I\{ F_i < t \} I \{ S_i \geq \theta t \}}{(1 - \theta t) \sum^m_{i = 1} I\{ F_i < t \}}  - \frac{ \sum^m_{i = 1} I\{ S_i < t \} I \{ \text{$H^{u/n}_i$ is true} \}}{\sum^m_{i = 1} I\{ F_i < t \}} \right) \\
    \geq \lim_{m \to \infty} \left(\frac{\sum^m_{i = 1} I\{ F_i < \delta \}}{m} \right) \inf_{t \geq \delta}  \left(  t \cdot  \frac{\sum^m_{i = 1} I\{ F_i < t \} I \{ S_i \geq \theta t \}}{(1 - \theta t) \sum^m_{i = 1} I\{ F_i < t \}} - \frac{ \sum^m_{i = 1} I\{ S_i < t \} I \{ \text{$H^{u/n}_i$ is true} \}}{\sum^m_{i = 1} I\{ F_i < t \}} \right) \\
    \geq \lim_{m \to \infty} \left( \frac{\sum^m_{i = 1} I\{ F_i < \delta \}}{m} \right) \inf_{t \geq \delta} \left( \frac{\pi_0 \tilde{S}_0(t) m}{\sum^m_{i=1} I \{ F_i < t \} } - \frac{ \sum^m_{i = 1} I\{ S_i < t \} I \{ \text{$H^{u/n}_i$ is true} \}}{\sum^m_{i = 1} I\{ F_i < t \}} \right) \\
    \geq \lim_{m \to \infty} \left(- \frac{\sum^m_{i = 1} I\{ F_i < \delta \}}{m} \right) \sup_{t \geq \delta}  \left| \frac{\pi_0 \tilde{S}_0(t) m}{\sum^m_{i=1} I \{ F_i < t \} } - \frac{ \sum^m_{i = 1} I\{ S_i < t \} I \{ \text{$H^{u/n}_i$ is true} \}}{\sum^m_{i = 1} I\{ F_i < t \}} \right| \\
    = 0
\end{align*}
where the second inequality is a result of \lemref{piasymp} and the last equality is a consequence of \lemref{piasymp2}.
\end{proof}

\subsection{Proof of \thmref{AdaBon_kFWER_control}}\label{app:AdaBon_kFWER_control}
\begin{proof}[\unskip\nopunct]
For any $t \in (0,1]$, we have by \assumpref{asymp} that 
\begin{align}
     \lim_{m \to \infty} t \cdot \frac{\sum^m_{i=1} I \{ F_i < t \} I\{ S_i \geq \theta t \} }{m (1 - \theta t)} &= \lim_{m \to \infty} t \cdot \frac{\sum^m_{i=1} I \{ F_i < t \} (1 - I\{ S_i < \theta t \}) }{m (1 - \theta t)} \nonumber \\
      &= \lim_{m \to \infty} t \cdot \frac{\sum^m_{i=1} I \{ F_i < t \} - \sum^m_{i=1} I \{ F_i < t \} I\{ S_i < \theta t \} }{m (1 - \theta t) } \nonumber \\
       &= \lim_{m \to \infty} t \cdot \frac{\sum^m_{i=1} I \{ F_i < t \} - \sum^m_{i=1} I\{ S_i < \theta t \} }{m (1 - \theta t) } \nonumber \\
       &= t \cdot \frac{(\pi_0 \tilde{F}_0(t) + \pi_1 \tilde{F}_1(t) ) - (\pi_0 \tilde{S}_0(\theta t) + \pi_1 \tilde{S}_1(\theta t) ) }{1 - \theta t} \nonumber \\
       &= t \cdot  \frac{\tilde{F}(t) - \tilde{S}(\theta t) }{1 - \theta t} \label{nice_res}
\end{align}
with probability 1, where $\tilde{F}(t)$ and $\tilde{S}(t)$ are as defined in \eqref{full_FS}. By \eqref{lim_cond} in \assumpref{asymp} and \eqref{nice_res}, there exists a $t' > 0$ such that
\begin{equation}\label{aw}
   \alpha \cdot \omega - t' \cdot \frac{\tilde{F}(t') - \tilde{S}(\theta t') }{1 - \theta t'} = \frac{\epsilon}{2}
\end{equation}
for some constant $\epsilon > 0$. For a sufficiently large $m$, we also have by \eqref{lim_cond} in \assumpref{asymp} and \eqref{nice_res} that
\begin{align*}
    \left| t' \cdot \frac{\tilde{F}(t') - \tilde{S}(\theta t') }{1 - \theta t'} - t' \cdot \frac{\sum^m_{i=1} I \{ F_i < t' \} I\{ S_i \geq \theta t' \} }{m (1 - \theta t')} \right|  < \frac{\epsilon}{2},
\end{align*}
which then implies from \eqref{aw} that 
\begin{align*}
    t' \cdot \frac{\sum^m_{i=1} I \{ F_i < t' \} I\{ S_i \geq \theta t' \} }{m (1 - \theta t')} < \alpha \cdot \omega \implies t' \cdot \frac{\sum^m_{i=1} I \{ F_i < t' \} I\{ S_i \geq \theta t' \} }{1 - \theta t'} < \alpha \cdot k.
\end{align*}
Then by the definition of $\hat{t}_{\theta}$ in \eqref{t_hat}, it follows that
\begin{align*}
    \lim_{m \to \infty} \hat{t}_{\theta} \geq t'
\end{align*}
with probability 1. The display above and \lemref{piasymp3} then gives the following result:
\begin{align}
    \lim_{m \to \infty} \inf \left( \hat{t}_{\theta} \cdot \frac{\sum^m_{i = 1} I\{ F_i < \hat{t}_{\theta} \} I \{ S_i \geq \theta \hat{t}_{\theta} \}}{m (1 - \theta \hat{t}_{\theta})} - \frac{\sum^m_{i = 1} I\{ S_i < \hat{t}_{\theta} \} I \{ \text{$H^{u/n}_i$ is true} \}}{m} \right) \nonumber \\
    \geq  \lim_{m \to \infty} \inf_{t \geq \delta} \left( t \cdot \frac{\sum^m_{i = 1} I\{ F_i < t \} I \{ S_i \geq \theta t \}}{m (1 - \theta t)} - \frac{\sum^m_{i = 1} I\{ S_i < t \} I \{ \text{$H^{u/n}_i$ is true} \}}{m} \right) \geq 0 \label{liminfres}
\end{align}
for $\delta = t'/2$. Since it holds for any $m$ that
\begin{equation*}
    \hat{t}_{\theta} \cdot \frac{\sum^m_{i = 1} I\{ F_i < \hat{t}_{\theta} \} I \{ S_i \geq \theta \hat{t}_{\theta} \}}{1 - \theta \hat{t}_{\theta}}  \leq k \cdot \alpha,
\end{equation*}
by the definition of $\hat{t}_{\theta}$ in $\eqref{t_hat}$, it must follow from \eqref{liminfres} that
\begin{align*}
    \lim_{m \to \infty} \inf \left(  \frac{k \cdot \alpha}{m } - \frac{\sum^m_{i = 1} I\{ S_i < \hat{t}_{\theta} \} I \{ \text{$H^{u/n}_i$ is true} \}}{m} \right) \geq 0. 
\end{align*}
We can therefore conclude that
\begin{align*}
    \lim_{m \to \infty} \sup \left( \frac{ \sum^m_{i = 1} I\{ S_i < \hat{t}_{\theta} \} I \{ \text{$H^{u/n}_i$ is true} \} }{m} \right) \leq \omega \cdot \alpha. 
\end{align*}
By Fatou's lemma, we have
\begin{align*}
    \limsup_{m \to \infty} \left( \mathbb{E} \left[ \frac{\sum^m_{i = 1} I\{ S_i < \hat{t}_{\theta} \} I \{ \text{$H^{u/n}_i$ is true} \} }{m}  \right] \right) &\leq    \mathbb{E} \left[ \limsup_{m \to \infty} \left( \frac{\sum^m_{i = 1} I\{ S_i < \hat{t}_{\theta} \} I \{ \text{$H^{u/n}_i$ is true} \} }{m} \right) \right] \\
    &\leq \omega \cdot  \alpha.
\end{align*}
It then follows from the above display and Markov's inequality that
\begin{align*}
    \limsup_{m \to \infty} \Pr \left( \text{$k$-FWER}(\mathcal{R}) \geq k \right) &= \limsup_{m \to \infty} \Pr \left( \sum^m_{i = 1} I\{ S_i < \hat{t}_{\theta} \} I \{ \text{$H^{u/n}_i$ is true} \} \geq k  \right) \\
    &\leq \limsup_{m \to \infty} \left( \mathbb{E} \left[ \frac{\sum^m_{i = 1} I\{ S_i < \hat{t}_{\theta} \} I \{ \text{$H^{u/n}_i$ is true} \} }{k }  \right] \right) \\
        &= \limsup_{m \to \infty} \left( \mathbb{E} \left[ \frac{\sum^m_{i = 1} I\{ S_i < \hat{t}_{\theta} \} I \{ \text{$H^{u/n}_i$ is true} \} }{m \cdot \omega }  \right] \right) \\
        &\leq \frac{\omega \cdot \alpha}{\omega} \\
        &= \alpha.
\end{align*}
\end{proof}

\section{Additional simulation results}\label{app:additional_sim_results}
Additional simulation results comparing the methods under $k$-FWER control for $k = 5$ and $10$, using the same $\pi_1$ and $\rho$ settings in \secref{sim_res}, are provided in \figref{sim_k5} and \figref{sim_k10} respectively. Note that methods $(e)$ and $(f)$ are designed for FWER control and are therefore omitted from these comparisons.

In these additional simulations, the observed $k$-FWER levels for $k = 5$ and $k = 10$ across all methods remain close to zero. AdaFilter-AdaBon exhibits the least conservative control, most prominently when $\rho = 0.8$ (large positive correlation) and $k = 5$. Regarding power, AdaFilter-AdaBon surpasses the other methods by a clear margin across all simulation settings, with the improvements over the Bonferroni and Hochberg procedures being especially pronounced when $u = 4$.

\begin{figure}
    \centering
    \includegraphics[width=0.85\linewidth]{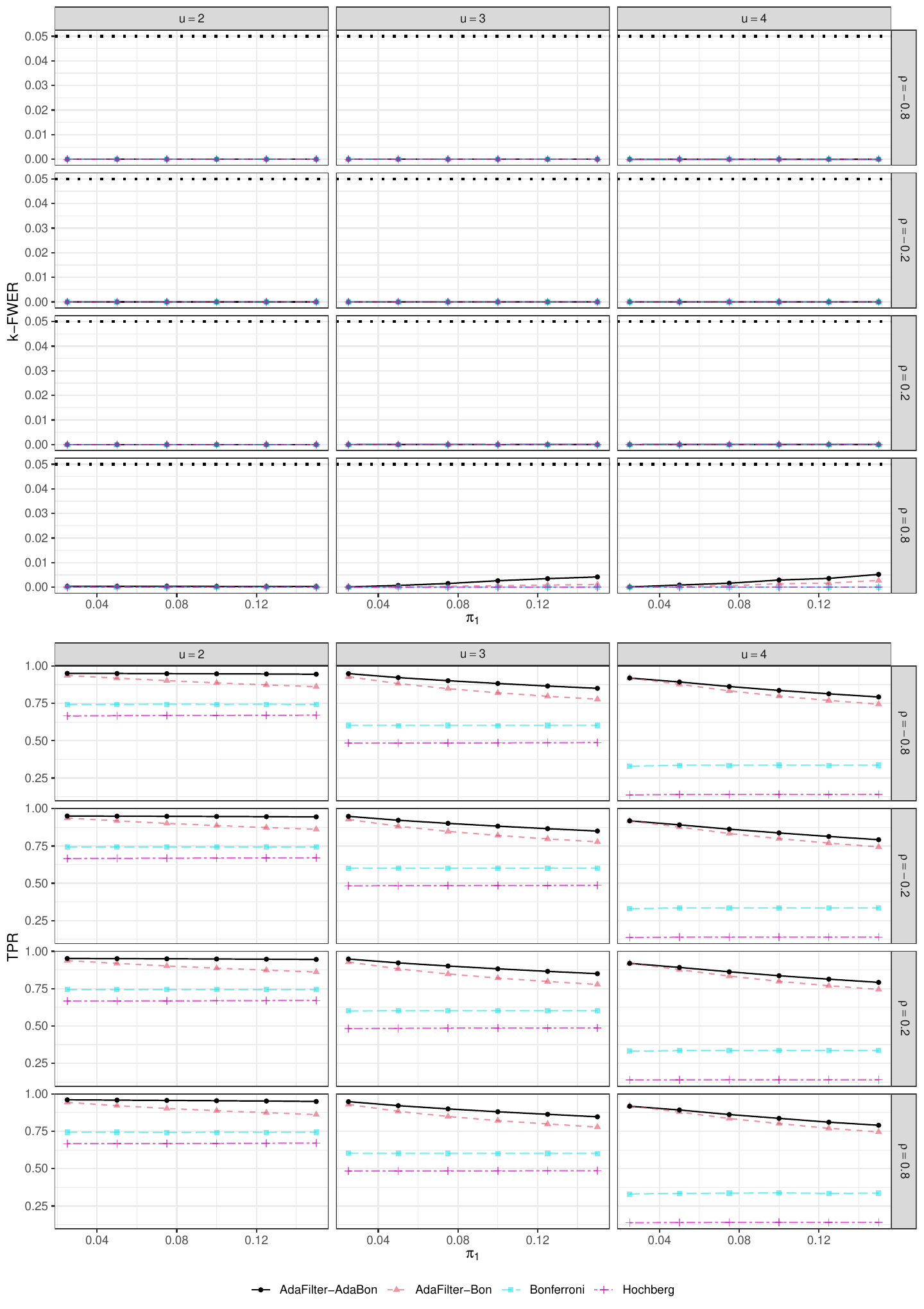}
    \caption{Empirical $k$-FWER and TPR levels of the compared methods for the parameter settings $k = 5$, $\alpha = 0.05$, $u \in \{ 2,3,4 \}$, $\pi_1 \in \{ 0.025, 0.05, 0.075, 0.10, 0.125, 0.15 \}$, and $\rho \in \{ -0.8, -0.2, 0.2, 0.8 \}$.}
    \label{fig:sim_k5}
\end{figure}

\begin{figure}
    \centering
    \includegraphics[width=0.85\linewidth]{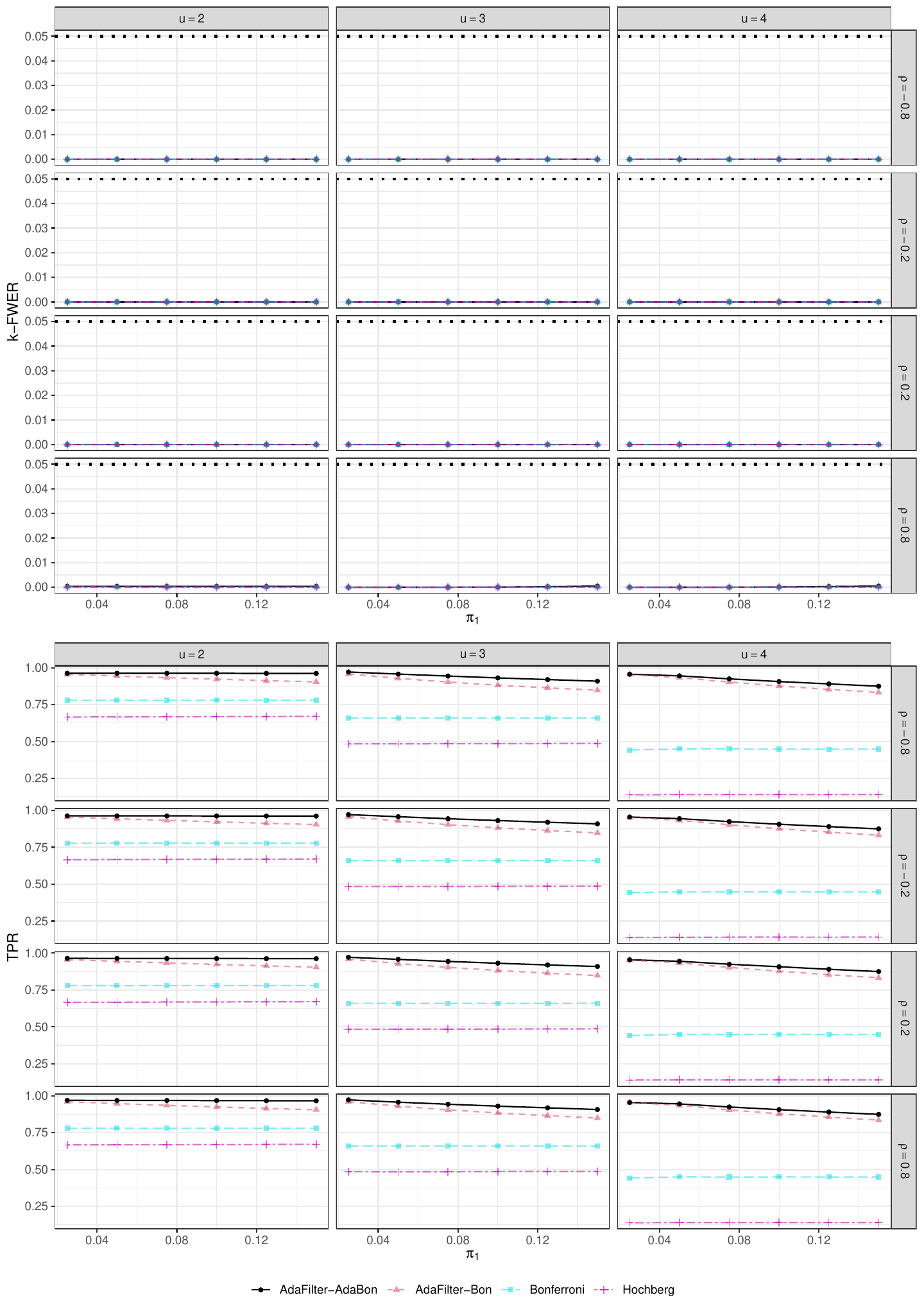}
    \caption{Empirical $k$-FWER and TPR levels of the compared methods for the parameter settings $k = 10$, $\alpha = 0.05$, $u \in \{ 2,3,4 \}$, $\pi_1 \in \{ 0.025, 0.05, 0.075, 0.10, 0.125, 0.15 \}$, and $\rho \in \{ -0.8, -0.2, 0.2, 0.8 \}$.}
    \label{fig:sim_k10}
\end{figure}

\section{Additional error control proofs}
\subsection{Proof of \corref{AdaBon_FDX_control}}\label{app:AdaBon_FDX_control}
\begin{proof}[\unskip\nopunct]
We have that
\begin{align*}
    \text{FDX}(\mathcal{R}) &= \Pr \left( \frac{ \sum^m_{i = 1} I \{ S_i \leq \hat{\tau} \} I \{ \text{$H^{u/n}_i$ is true} \}}{1 \vee \sum^m_{i = 1} I \{ S_i \leq \hat{\tau} \}} \geq \gamma \right) \\
    &\leq 
    \Pr \left( \frac{\sum^m_{i = 1} I \{ S_i \leq \hat{\tau} \} I \{ \text{$H^{u/n}_i$ is true} \}} {1 \vee\sum^m_{i = 1} I \{ S_i \leq \hat{\tau} \}} \geq \frac{\sum^m_{i=1} I \{ \hat{t}_{\theta} \leq S_i \leq \hat{\tau} \} + k}{1 \vee \sum^m_{i = 1} I \{ S_i \leq \hat{\tau} \}} \right) \\
    &= \Pr \left( \sum^m_{i = 1} I \{ S_i \leq \hat{\tau} \} I \{ \text{$H^{u/n}_i$ is true} \} \geq \sum^m_{i=1} I \{ \hat{t}_{\theta} \leq S_i \leq \hat{\tau} \} + k \right) \\
    &\leq \Pr \left( \sum^m_{i = 1} I \{ S_i \leq \hat{\tau} \} I \{ \text{$H^{u/n}_i$ is true} \} \geq \sum^m_{i = 1} I \{ \hat{t}_{\theta} \leq S_i \leq \hat{\tau} \} I \{ \text{$H^{u/n}_i$ is true} \} + k \right) \\
    &\leq \Pr \left( \sum^m_{i = 1} I \{ S_i < \hat{t}_{\theta} \} I \{ \text{$H^{u/n}_i$ is true} \} \geq  k \right) \\
    &= k\text{-FWER}(\mathcal{R})
\end{align*}
where the first inequality follows from \eqref{hat_tau}. Given the display above and \thmref{AdaBon_kFWER_control}, we have that
\begin{equation}\label{FDX_bound}
  \limsup_{m \longrightarrow \infty} \text{FDX}(\mathcal{R}) \leq   \limsup_{m \longrightarrow \infty}  k\text{-FWER}(\mathcal{R})  \leq \alpha,
\end{equation}
which proves FDX control. To prove FDR control, note that
\begin{align*}
    \text{FDR}(\mathcal{R}) &= \mathbb{E} \left[ \frac{ \sum^m_{i = 1} I \{ S_i \leq \hat{\tau} \} I \{ \text{$H^{u/n}_i$ is true} \}}{1 \vee \sum^m_{i = 1} I \{ S_i \leq \hat{\tau} \}} \right] \\
    &\leq \mathbb{E} \left[ \frac{ \sum^m_{i = 1} I \{ S_i \leq \hat{\tau} \} I \{ \text{$H^{u/n}_i$ is true} \}}{1 \vee \sum^m_{i = 1} I \{ S_i \leq \hat{\tau} \}} \cdot I \left\{ \frac{ \sum^m_{i = 1} I \{ S_i \leq \hat{\tau} \} I \{ \text{$H^{u/n}_i$ is true} \}}{1 \vee \sum^m_{i = 1} I \{ S_i \leq \hat{\tau} \}} \leq \gamma \right\} \right] \\
    &+ \mathbb{E} \left[ \frac{ \sum^m_{i = 1} I \{ S_i \leq \hat{\tau} \} I \{ \text{$H^{u/n}_i$ is true} \}}{1 \vee \sum^m_{i = 1} I \{ S_i \leq \hat{\tau} \}} \cdot I \left\{ \frac{ \sum^m_{i = 1} I \{ S_i \leq \hat{\tau} \} I \{ \text{$H^{u/n}_i$ is true} \}}{1 \vee \sum^m_{i = 1} I \{ S_i \leq \hat{\tau} \}} \geq  \gamma \right\} \right] \\
   &\leq \gamma + \mathbb{E} \left[ I \left\{ \frac{ \sum^m_{i = 1} I \{ S_i \leq \hat{\tau} \} I \{ \text{$H^{u/n}_i$ is true} \}}{1 \vee \sum^m_{i = 1} I \{ S_i \leq \hat{\tau} \}} \geq  \gamma \right\} \right] \\
    &= \gamma + \Pr \left(\frac{ \sum^m_{i = 1} I \{ S_i \leq \hat{\tau} \} I \{ \text{$H^{u/n}_i$ is true} \}}{1 \vee \sum^m_{i = 1} I \{ S_i \leq \hat{\tau} \}} \geq  \gamma \right) \\
    &= \gamma + \text{FDX}(\mathcal{R}).
\end{align*}
The result then immediately follows from \eqref{FDX_bound}.
\end{proof}

\end{document}